\documentclass[letter,11pt]{article}
\usepackage{float,graphicx,verbatim,fullpage,hyperref,amssymb,amsmath,amsthm,amsfonts,enumerate,multicol, xspace,comment,xcolor,thm-restate,url, cite}
\usepackage[T1]{fontenc}
\usepackage[margin=1in]{geometry}
\usepackage{subfig}
\usepackage[margin=1in]{geometry}
\newcommand*{\email}[1]{\texttt{#1}}
\usepackage{hyperref}
\usepackage{fancyhdr}

\newtheorem{theorem}{Theorem}[section]
\newtheorem{lemma}{Lemma}[section]
\newtheorem{corollary}{Corollary}[section]

\newtheorem{proposition}{Proposition}[section]
\newtheorem{claim}{Claim}[section]
\newtheorem{remark}{Remark}[section]

\def\final{0}  
\def\iflong{\iffalse}
\ifnum\final=0  
\newcommand{\cnote}[1]{{\color{blue}[{\tiny Calvin: \bf #1}]\marginpar{*}}}
\newcommand{\wnote}[1]{{\color{cyan}[{\tiny Weihang: \bf #1}]\marginpar{*}}}
\newcommand{\knote}[1]{{\color{red}[{\tiny Karthik: \bf #1}]\marginpar{\color{red}*}}}
\newcommand{\todo}[1]{{\color{red}[{\tiny TODO: \bf #1}]\marginpar{\color{red}*}}}
\else 
\newcommand{\cnote}[1]{}
\newcommand{\wnote}[1]{}
\newcommand{\knote}[1]{}
\newcommand{\todo}[1]{}
\fi  

\newcommand{\R}{\mathbb{R}}

\newcommand{\deltacard}{d}

\def\complement#1{\overline{#1}}
\def\set#1{\{ #1 \}}
\usepackage[utf8]{inputenc}

\title{Approximate minimum cuts and their enumeration\footnote{University of Illinois, Urbana-Champaign. Email: \email{\{calvinb2, karthe, weihang3\}@illinois.edu}. Supported in part by NSF grants CCF-1814613 and  CCF-1907937.}
}
\author{
Calvin Beideman \and 
Karthekeyan Chandrasekaran \and 
Weihang Wang
}
\date{}

\begin{document}
\maketitle

\begin{abstract}
We show that every $\alpha$-approximate minimum cut in a connected graph is the \emph{unique} minimum $(S,T)$-terminal cut for some subsets $S$ and $T$ of vertices each of size at most $\lfloor2\alpha\rfloor+1$. This leads to an alternative proof that the number of $\alpha$-approximate minimum cuts in a $n$-vertex connected graph is $n^{O(\alpha)}$ and they can all be enumerated in deterministic polynomial time for constant $\alpha$. 
\end{abstract}

\section{Introduction}

Let $G=(V,E)$ be a graph with positive edge costs $c:E\rightarrow \R_+$. 
A partitioning of the vertex set into $2$ non-empty parts is known as a \emph{cut}. 
For a non-empty proper subset $U$ of vertices,
we will use 
$\complement{U}$ to denote $V\setminus U$, 
$(U, \complement{U})$ to denote the cut associated with $U$, 
$\delta(U)$
to denote the set of edges crossing the cut $(U, \complement{U})$,  
and $\deltacard(U):=\sum_{e\in \delta(U)}c(e)$ to denote the cut value of $U$. 
We will use $\lambda$ to denote the value of a minimum cut in $G$, i.e., 
\[
\lambda:=\min\{d(U):\emptyset\neq U\subsetneq V\}.
\]
For $\alpha\geq 1$, a cut $(U,\complement{U})$ is an \emph{$\alpha$-approximate minimum cut} if $d(U)\leq\alpha\lambda$. 
A fundamental graph structural result states that the number of $\alpha$-approximate minimum cuts in a $n$-vertex connected graph is $O(n^{2\alpha})$. This structural result forms the backbone of several algorithmic and representation results in graphs---e.g., fast randomized construction of skeletons \cite{Karger00}, existence of cut sparsifiers which in turn finds applications in streaming and sketching \cite{AG09, AGM12, AGM12-pods, KK15},  approximation algorithms for TSP \cite{Zen19, Ben95, BG08, Benthesis}, computing reliability of probabilistic networks \cite{Kar99}, and polynomial time algorithms for connectivity augmentation \cite{NGM97}. 
The structural result along with a randomized polynomial-time algorithm to enumerate all constant-approximate minimum cuts was first shown via Karger's random contraction technique \cite{KS96}. 
Subsequently, the splitting-off technique and the tree packing technique have been used to show the structural result along with a \emph{deterministic} polynomial-time algorithm to enumerate all constant-approximate minimum cuts \cite{NNI97, CQX20}. 


In this work, we give a fourth technique to bound the number of $\alpha$-approximate minimum cuts by $n^{O(\alpha)}$ along with a deterministic polynomial-time algorithm to enumerate all constant-approximate minimum cuts. 
Let $S$, $T$ be
disjoint non-empty subsets of vertices. A $2$-partition $(U, \complement{U})$ is
an $(S,T)$-terminal cut if $S\subseteq U\subseteq V\setminus T$. Here,
the set $U$ is known as the \emph{source set} and the set $\complement{U}$ is
known as the \emph{sink set}.  An $(S,T)$-terminal cut with minimum cut value will be denoted 
as a \emph{minimum $(S,T)$-terminal cut}. The following is our main result. 

\begin{restatable}{theorem}{thmStructure}
\label{thm: main}
Let $G=(V,E)$ be a connected graph with positive edge costs and 
$(U, \complement{U})$ 
be an $\alpha$-approximate minimum cut for some $\alpha \ge 1$.
Then, there exist subsets $S,T\subseteq V$ with $|S|,|T|\leq \lfloor 2\alpha\rfloor+1$ such that $(U,\complement{U})$ is the unique minimum $(S,T)$-terminal cut.
\end{restatable}

In other words, every $\alpha$-approximate minimum cut $(U, \complement{U})$ in a connected graph can be recovered as the \emph{unique} minimum $(S, T)$-terminal cut for some subsets $S$ and $T$ of sizes at most $\lfloor 2\alpha\rfloor+1$. 
A few remarks are in order. 

\begin{remark}
An immediate consequence of Theorem \ref{thm: main} is that the number of distinct $\alpha$-approximate minimum cuts in a $n$-vertex connected graph is at most $n^{4\alpha+2}$ and they can all be enumerated using $n^{4\alpha+2}$ minimum $(S, T)$-terminal cut computations. We recall that minimum $(S, T)$-terminal cut in a given graph with given subsets $S$ and $T$ of vertices can be computed in deterministic polynomial time. 
\end{remark}

\begin{remark}
The size bound for subsets $S$ and $T$ in Theorem \ref{thm: main} is tight up to an additive factor of one: Consider the $n$-vertex cycle graph $G$ with $\alpha$ being a positive integer where $n\ge 4\alpha$ and with all edge costs being $1$. Let $(U,\complement{U})$ be  
a cut in $G$ with $d(U)=2\alpha$ and moreover, each connected component in $G-\delta(U)$ has at least $2$ vertices (see Figure \ref{fig:tight example} for such an example where $\alpha=4$). For such a cut $(U, \complement{U})$, the only choice of $S$ and $T$ for which 
$(U,\complement{U})$ is the \emph{unique} minimum $(S,T)$-terminal cut is given by
\begin{align*}
    S&:=\{v\in U:v\text{ is an end vertex of an edge in }\delta(U)\} \text{ and}
    \\T&:=\{v\in \complement{U}:v\text{ is an end vertex of an edge in }\delta(U)\}.
\end{align*}
Hence, this example shows that both $S$ and $T$ need to have at least 
$2\alpha$ vertices in order to recover $(U, \complement{U})$ as the \emph{unique} minimum $(S,T)$-terminal cut. 
\end{remark}

\begin{figure}[H]
    \centering
    \includegraphics[width=0.4\textwidth]{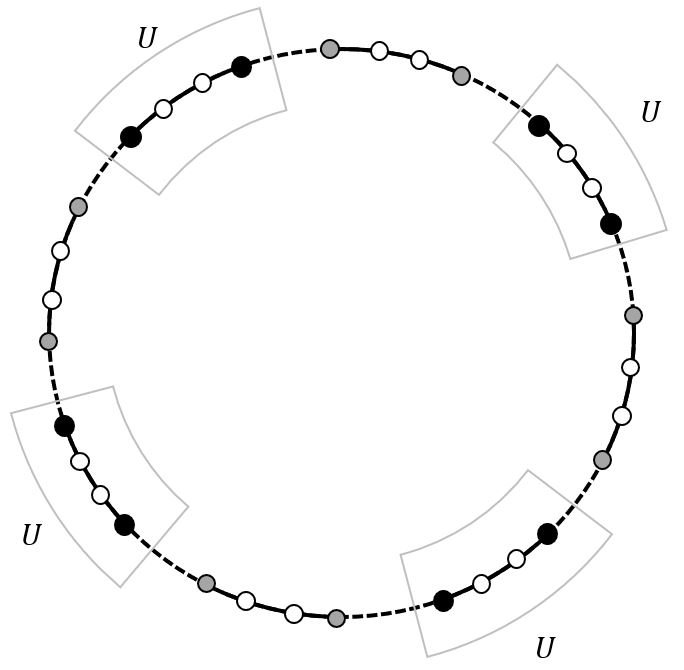}
    \caption{An example showing the tightness of Theorem \ref{thm: main}: The set $U$ consists of vertices within all gray boxes and the set $\complement{U}$ consists of all vertices outside gray boxes; the dotted edges correspond to the edges in the cut set $\delta(U)$. The only choice of $S$ and $T$ for which $(U, \complement{U})$ is the unique minimum $(S, T)$-terminal cut is as follows: $S$ consists of all vertices that are shaded black and $T$ consists of all vertices that are shaded gray.}
    \label{fig:tight example}
\end{figure}

Theorem \ref{thm: main} is inspired by a result of Goemans and Ramakrishnan \cite{GR95}. They showed the following structural theorem for $(4/3-\epsilon)$-approximate minimum cuts: fix an arbitrary vertex $t$ in a  connected graph $G=(V,E)$ and let $(U, \complement {U})$ be a $(4/3-\epsilon)$-approximate minimum cut with $t\in \complement{U}$ for some $\epsilon >0$. Then, there exists a subset $S$ of size at most $2$ such that $(U, \complement{U})$ is the unique minimum $(S,\set{t})$-terminal cut. 
Their result leads to the following natural question: for every $\alpha \ge 1$, can every $\alpha$-approximate minimum cut $(U, \complement{U})$ be obtained as a minimum $(S, T)$-terminal cut for some subsets $S$ and $T$ of vertices each of size at most some function of $\alpha$? 
Our Theorem \ref{thm: main} answers this question affirmatively. However, our proof of Theorem \ref{thm: main} deviates significantly from the proof approach of Goemans-Ramakrishnan's structural result for $(4/3-\epsilon)$-approximate minimum cuts. 
Their proof proceeds via contradiction and relies on the \emph{submodular triple inequality} which holds for the graph cut function. The submodular triple inequality shows that $3$ non-empty sets satisfying certain crossing properties can be uncrossed to obtain $4$ disjoint non-empty sets without increasing the sum of their cut values; thus, if all $3$ sets have cut value less than $(4/3)\lambda$, then one of the $4$ sets should have cut value strictly smaller than $\lambda$, a contradiction to the definition of $\lambda$. However, there does not seem to be a generalization of the submodular triple inequality to larger number of sets. Our proof of Theorem \ref{thm: main} circumvents the submodular triple inequality but achieves its intended purpose in this context---we show that a sufficiently large number of $\alpha$-approximate minimum cuts with certain crossing properties can be uncrossed to obtain a cut of value cheaper than $\lambda$, thus leading to a contradiction again. 


\subsection{Preliminaries}



Let $G=(V,E)$ be a graph with positive edge costs $c:E\rightarrow \R_+$. 
For disjoint non-empty subsets $S,T\subseteq V$, there can be multiple minimum $(S,T)$-terminal cuts. We will be interested in \emph{source
  minimal} minimum $(S,T)$-terminal cuts. A minimum $(S,T)$-terminal cut $(U,\complement{U})$ is a source
  minimal minimum $(S,T)$-terminal cut if for all minimum $(S,T)$-terminal cut $(X,\complement{X})$, we have $U\subseteq X$.
For every pair of disjoint non-empty subsets $S$ and $T$ of vertices, there exists a unique source minimal minimum $(S,T)$-terminal cut (this is due to submodularity---e.g., see \cite{GR95}).

We recall that the graph cut function $d:2^V\rightarrow \R$ is symmetric and submodular, i.e., $d(A) = d(V\setminus A)$ and $d(A) + d(B) \ge d(A\cap B) + d(A\cup B)$ for all $A, B\subseteq V$. We will need an uncrossing result that relies on more careful counting of edges than simply employing the submodularity inequality. 
We begin with some notation that will help in such careful counting.
Let $(Y_1, \ldots, Y_p, W, Z)$
be a partition of $V$ into $p+2$ non-empty parts. 
We define 
\begin{align*}
    \sigma(Y_1, \ldots, Y_p, W, Z)&:=2\left(\sum_{uv\in E:\ u\in Y_i,\ v\in Y_j\text{ for distinct }i,\ j\in[p]} c(uv)+\sum_{uv\in E:\ u\in W,\ v\in Z}c(uv)\right)\\
    &\quad \quad \quad \quad \quad \quad \quad \quad +\sum_{uv\in E:\ u\in \cup_{i\in[p]}Y_i,\ v\in W\cup Z}c(uv)
\end{align*}
We note that
\begin{align}
    \sum_{i\in[p]}d(Y_i) &\le \sigma(Y_1, \ldots, Y_p, W, Z). \label{eq:sigma-lower-bound}
\end{align}
The following lemma uncrosses a collection of $p$ sets to obtain a partition of the vertex set into $p+2$ non-empty parts with small $\sigma$-value. 
The lemma was shown by Kamidoi, Yoshida, and Nagamochi in the context of minimum $k$-cuts \cite{KYN07}. It was also extended to hypergraphs by Chandrasekaran and Chekuri \cite{CC20}. See Figure \ref{figure:uncrossing} for an illustration of the sets that appear in the statement of the lemma. 
\begin{lemma}\label{lemma:uncrossing} \cite{KYN07, CC20}
  Let $G=(V,E)$ be a graph with positive edge costs and
  $\emptyset\neq R\subsetneq U\subsetneq V$. Let
  $S=\{u_1,\ldots, u_p\}\subseteq U\setminus R$ for $p\ge 2$. For each $i\in [p]$, let
  $(\complement{A_i}, A_i)$ be a minimum
  $((S\cup R)\setminus \set{u_i}, \complement{U})$-terminal
  cut. Suppose that
  $u_i\in A_i\setminus (\cup_{j\in [p]\setminus \set{i}}A_j)$ for
  every $i\in [p]$.  Let
\[
Z:= \cap_{i=1}^p \complement{A_i},\ W:= \cup_{1\le i<j\le p}(A_i \cap A_j),\ \text{and}\ Y_i:=A_i-W\ \forall i\in [p].
\]
Then, $(Y_1, \ldots, Y_p, W, Z)$ is a partition of $V$ into $p+2$ non-empty parts with
\[
\sigma(Y_1, \ldots, Y_p, W, Z) \le \min\{\deltacard(A_i) + \deltacard(A_j): i, j\in [p], i\neq j\}.
\]
\end{lemma}

\begin{figure}[htb]
\centering
\includegraphics[width=0.6\textwidth]{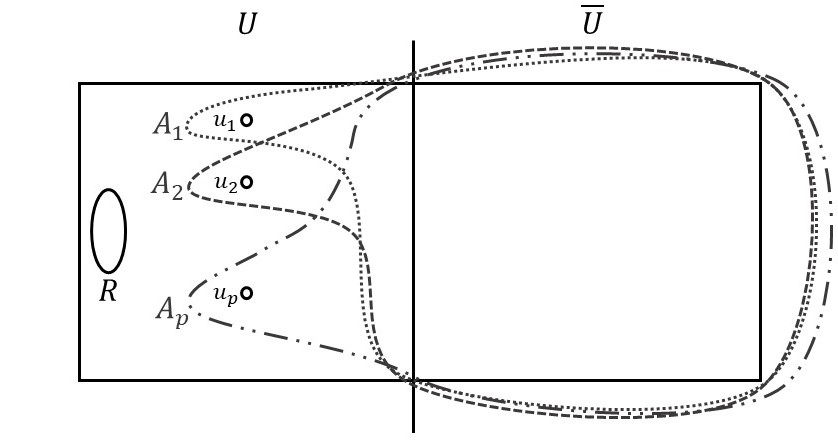}
\caption{Illustration of the sets that appear in the statement of Lemma  \ref{lemma:uncrossing}.} 
\label{figure:uncrossing}
\end{figure}

We briefly remark on the lemma: known proofs of the lemma are by induction on $p$ with an uncrossing argument that involves careful edge counting. The use of $\sigma(Y_1, \ldots, Y_p, W, Z)$ instead of $\sum_{i\in [p]}d(Y_i)$ on the LHS is crucial for the induction argument---this is why we prefer to state the lemma as giving an upper bound on $\sigma(Y_1, \ldots, Y_p, W, Z)$ although we will use it only to conclude an upper bound on $\sum_{i\in [p]}d(Y_i)$ via inequality \eqref{eq:sigma-lower-bound}. 

\section{Proof of Theorem \ref{thm: main}}

We will derive Theorem \ref{thm: main} from the following theorem.

\begin{theorem}\label{thm: SU mincut}
Let $G=(V,E)$ be a connected graph with positive edge costs and $(U, \complement{U})$ be an $\alpha$-approximate minimum cut for some $\alpha \ge 1$.
Then, there exists a subset $S\subseteq U$ with $|S|\leq \lfloor 2\alpha\rfloor+1$ such that $(U,\complement{U})$ is the unique minimum $(S,\complement{U})$-terminal cut.
\end{theorem}

\begin{proof}
Consider the collection $\mathcal{C}:=\{Q\subsetneq V:\complement{U} \subsetneq Q ,d(Q)\leq d(U)\}$. Suppose $\mathcal{C}$ is empty. Let $S = \{x\}$ for some arbitrary $x \in U$ and let $(X,\complement{X})$ be a minimum $(S,\complement{U})$-terminal cut. Then,  $\complement{U} \subseteq \complement{X} \subsetneq V$ and $d(\complement{X}) \leq d(U)$. Since $\complement{X} \not\in \mathcal{C}$ (because $\mathcal{C}$ is empty), we must have $\complement{X} = \complement{U}$. Therefore, $(U, \complement{U})$ is the unique minimum $(S,\complement{U})$-terminal cut.

Next, suppose $\mathcal{C} \neq \emptyset$. Let $S\subseteq U$ be a minimal transversal of the collection $\mathcal{C}$, i.e., $S\subseteq U$ is a minimal set with $S\cap Q\neq\emptyset$ for all $Q\in\mathcal{C}$. Proposition \ref{prop: SU mincut} and Lemma \ref{lem: S size bound} complete the proof of Theorem \ref{thm: SU mincut}.
\end{proof}

\begin{proposition}\label{prop: SU mincut}
$(U,\complement{U})$ is the unique minimum $(S,\complement{U})$-terminal cut.
\end{proposition}

\begin{proof}
Suppose $(X,\complement{X})$ be the source minimal minimum $(S,\complement{U})$-terminal cut with $X\neq U$. Since $(U,\complement{U})$ is a $(S,\complement{U})$-terminal cut, we have $d(\complement{X})=d(X)\leq d(U)$. Moreover, we have $\emptyset\neq X\subsetneq U$, which implies that  $\complement{U}\subsetneq\complement{X}\subsetneq V$. Hence, we have $\complement{X}\in\mathcal{C}$, and thus $S\cap \complement{X}\neq\emptyset$. This contradicts with the assumption that $(X,\complement{X})$ is a $(S,\complement{U})$-terminal cut. Therefore, we must have $X=U$.
\end{proof}

\begin{lemma} \label{lem: S size bound}
$|S|\leq \lfloor 2\alpha\rfloor+1$.
\end{lemma}
\begin{proof}

For the sake of contradiction, suppose that $|S|\geq\lfloor 2\alpha\rfloor+2$. 
Let $\lambda$ denote the value of a minimum cut in $G$. 
Our proof strategy is to arrive at a cut with value cheaper than $\lambda$, thus contradicting the definition of $\lambda$. 
For convenience, we will denote $p:=|S|$ and write $S=\{u_1,\ldots, u_p\}$. For each $i\in[p]$, let $(\complement{A_i},A_i)$ be the source minimal minimum $(S-u_i,\complement{U})$-terminal cut. 

\begin{claim}\label{claim:min-separating-cuts-miss-u_i}
For every $i\in[p]$, we have $d(A_i)\leq d(U)$ and $u_i\in A_i$.
\end{claim}

\begin{proof}
Let $i\in [p]$. Since $S$ is a minimal transversal of the collection $\mathcal{C}$, there exists a set $B_i\in \mathcal{C}$ such that $B_i\cap S = \{u_i\}$. Hence, $(\complement{B_i},B_i)$ is a $(S-u_i,\complement{U})$-terminal cut. Therefore, 
\[
d(A_i)\leq d(B_i)\leq d(U). 
\]

We will show that $A_i$ is in the collection $\mathcal{C}$. 
By definition, $A_i\subseteq V\setminus (S-u_i)\subsetneq V$ and $\complement{U}\subseteq A_i$. If $A_i=\complement{U}$, then the above inequalities are equations (since $d(A_i) = d(\complement{U})=d(U)$) implying that $(\complement{B_i}, B_i)$ is a minimum  $(S-u_i, \complement{U})$-terminal cut, and consequently, $(\complement{B_i}, B_i)$ contradicts source minimality of the minimum $(S-u_i, \complement{U})$-terminal cut $(\complement{A_i},A_i)$. Therefore, $\complement{U}\subsetneq A_i$. Hence, $A_i$ is in the collection $\mathcal{C}$. 

We recall that the set $S$ is a transversal for the collection $\mathcal{C}$ and moreover, none of the elements of $S-u_i$ are in $A_i$  by definition of $A_i$. Therefore, the vertex $u_i$ must be in $A_i$.  
\end{proof}

We recall our assumption that $p\geq\lfloor 2\alpha\rfloor+2$. 
We note that $p-1>2\alpha \ge 2$. Using Claim \ref{claim:min-separating-cuts-miss-u_i}, we observe that the sets $U$, $R:=\{u_p\}$, $S=\set{u_1,\ldots,u_{p-1}}$, and partitions $(\complement{A_i}, A_i)$ for $i\in [p-1]$ satisfy the conditions of Lemma \ref{lemma:uncrossing}. Therefore, applying Lemma \ref{lemma:uncrossing} shows the existence of a partition $(Y_1, \ldots, Y_{p-1},W,Z)$ of $V$ into $p+1$ non-empty parts such that 
\begin{align*}
\sigma(Y_1, \ldots, Y_{p-1},W,Z)
&\leq\min\set{\deltacard(A_i) + \deltacard(A_j):i, j\in [p-1], i\neq j}
\le 2\deltacard(U)
\leq 2\alpha\lambda.
\end{align*}
By inequality \eqref{eq:sigma-lower-bound}, this implies that 
\begin{align*}
    \sum_{i\in[p-1]}d(Y_i)\leq 2\alpha\lambda.
\end{align*}
Therefore, there exists $i\in[p-1]$ such that
\begin{align*}
    d(Y_i)\leq \frac{2\alpha\lambda}{p-1}<\lambda.
\end{align*}
The last inequality is because $p-1>2\alpha$ and $\lambda>0$ since the graph is connected. Thus, we have a cut $(Y_i, \complement{Y_i})$ with cut value smaller than $\lambda$, a contradiction to the definition of $\lambda$. 
\end{proof}

Applying Theorem \ref{thm: SU mincut} to $(\complement{U}, U)$ yields the following corollary.

\begin{corollary}\label{coro: main}
Let $G=(V,E)$ be a connected graph with positive edge costs and $(U, \complement{U})$ be an $\alpha$-approximate minimum cut for some $\alpha \ge 1$.
Then, there exists a subset $T\subseteq \complement{U}$ with $|T|\leq \lfloor 2\alpha\rfloor+1$ such that $(U,\complement{U})$ is the unique minimum $(U,T)$-terminal cut.
\end{corollary}

We now restate Theorem \ref{thm: main} and prove it using Theorem \ref{thm: SU mincut} and Corollary \ref{coro: main}.

\thmStructure*
\begin{proof}
By Theorem \ref{thm: SU mincut}, there exists a subset $S\subseteq U$ with $|S|\leq \lfloor 2\alpha\rfloor+1$ such that $(U,\complement{U})$ is the unique minimum $(S,\complement{U})$-terminal cut. By Corollary \ref{coro: main}, there exists a subset $T\subseteq \complement{U}$ with $|T|\leq \lfloor 2\alpha\rfloor+1$ such that $(U,\complement{U})$ is the unique minimum $(U,T)$-terminal cut. For these choices of subsets $S$ and $T$, we now show that 
$(U,\complement{U})$ is the unique minimum $(S,T)$-terminal cut.

Let $(Y, \complement{Y})$ be an arbitrary minimum $(S,T)$-terminal cut. It suffices to show that $Y=U$. 
Since $(U, \complement{U})$ is a $(S,T)$-terminal cut, we have that
\[
d(Y) \le d(U). 
\]
Since $(Y\cap U, \complement{Y\cap U})$ is a $(S, \complement{U})$-terminal cut and $(U, \complement{U})$ is a minimum $(S, \complement{U})$-terminal cut, we have that 
\[
d(U)\le d(Y\cap U) . 
\]
Since $(Y\cup U, \complement{Y\cup U})$ is a $(U, T)$-terminal cut and $(U, \complement{U})$ is a minimum $(U, T)$-terminal cut, we have that 
\[
d(U) \le d(Y\cup U). 
\]
Using the above inequalities in conjunction with the submodularity of the cut function, we obtain that 
\[
2d(U)\le d(U\cap Y) + d(U\cup Y) \le d(U)+d(Y) \le 2 d(U).
\]
Hence, all of the above inequalities should be equations. Consequently, 
$(Y\cap U, \complement{Y\cap U})$ is a minimum $(S, \complement{U})$-terminal cut, and  
$(Y\cup U, \complement{Y\cup U})$ is a minimum $(U, T)$-terminal cut. This implies that $Y\cap U= U$ since $(U, \complement{U})$ is the unique minimum $(S, \complement{U})$-terminal cut and $Y\cup U=U$ since $(U, \complement{U})$ is the unique minimum $(U, T)$-terminal cut. Thus, we have that $Y=U$. 
\end{proof}

\section{An Open Problem}
The fundamental structural result that the number of $\alpha$-approximate minimum cuts in a $n$-vertex connected graph is at most $n^{O(\alpha)}$ has been extended to $r$-rank hypergraphs (the \emph{rank} of a hypergraph is the size of the largest hyperedge): Kogan and Krauthgamer \cite{KK15} used the random contraction technique to show that the 
number of $\alpha$-approximate minimum cuts in a $r$-rank $n$-vertex connected hypergraph is $O(2^{\alpha r}n^{2\alpha})$ and they can all be enumerated in randomized polynomial time for constant $\alpha$ and $r$. It would be interesting to give a \emph{deterministic} proof of this bound. 

\bibliographystyle{amsplain}
\bibliography{references}

\end{document}